\documentclass[a4paper,UKenglish,cleveref, autoref, thm-restate]{oasics-v2021}

\title{Visualization of Event Graphs for Train Schedules}
\titlerunning{Visualization of Event Graphs for Train Schedules}

  \author{Johann Hartleb}{DB InfraGO AG, Germany}{johann.hartleb@deutschebahn.com}{https://orcid.org/0000-0001-8101-1542}{}
  \author{Marie Schmidt}{Universität Würzburg, Germany \and \url{https://www.informatik.uni-wuerzburg.de/en/algo/team/schmidt-marie/}}{marie.schmidt@uni-wuerzburg.de}{https://orcid.org/0000-0001-9563-9955}{}
  \author{Samuel Wolf}{Universität Würzburg, Germany \and \url{https://www.informatik.uni-wuerzburg.de/en/algo/team/wolf-samuel/}}{samuel.wolf@uni-wuerzburg.de}{https://orcid.org/0009-0009-7098-6147}{}
  \author{Alexander~Wolff}{Universität Würzburg, Germany \and \url{https://www.informatik.uni-wuerzburg.de/en/algo/team/wolff-alexander/}}{}{https://orcid.org/0000-0001-5872-718X}{}
  \authorrunning{Johann Hartleb, Marie Schmidt, Samuel Wolf, and Alexander Wolff}

\Copyright{Johann Hartleb, Marie Schmidt, Samuel Wolf, and Alexander Wolff} 



\nolinenumbers 

\EventEditors{Jonas Sauer and Marie Schmidt}
\EventNoEds{2}
\EventLongTitle{25th Symposium on Algorithmic Approaches for Transportation Modelling, Optimization, and Systems (ATMOS 2025)}
\EventShortTitle{ATMOS 2025}
\EventAcronym{ATMOS}
\EventYear{2025}
\EventDate{September 18--19, 2025}
\EventLocation{Warsaw, Poland}
\EventLogo{}
\SeriesVolume{137}
\ArticleNo{14}
\hideOASIcs

\usepackage[color=blue!30, textsize=scriptsize]{todonotes}
\usepackage{optidef}
\usepackage{cite}
\usepackage[linesnumbered,algoruled,longend,vlined]{algorithm2e}
\usepackage[disableredefinitions]{complexity}

\DontPrintSemicolon
\SetArgSty{}
\SetKw{KwOr}{or}
\SetKw{KwAnd}{and}
\SetKw{KwNot}{not}
\SetKw{KwTrue}{true}
\SetKw{KwFalse}{false}
\SetKwInput{Input}{Input}
\SetKwInput{Output}{Output}
\SetFuncArgSty{textit}

\SetCommentSty{mycommfont}

\newcommand{\locationOf}[1]{\ensuremath{\ell(#1)}}

\newcommand{\trainOf}[1]{\ensuremath{\operatorname{train}(#1)}}
\newcommand{\timeOf}[1]{\ensuremath{t(#1)}}

\newcommand{\Polytime}{\cP}
\newclass{\paraNP}{para\textsf{-}NP}

\newcommand{\refilp}[1]{\ensuremath{\mathsf{ILP\ref{#1}}}}
\newtheorem{problem}[theorem]{Problem}

\theoremstyle{definition}
\newtheorem{reductionrule}{Reduction Rule}
\crefname{reductionrule}{Reduction Rule}{Reduction Rules}
\Crefname{reductionrule}{Reduction Rule}{Reduction Rules}
\crefname{problem}{Problem}{Problems}

\definecolor{defblue}{rgb}{0.121,0.47,0.705}
\let\emph\relax
\DeclareTextFontCommand{\emph}{\color{defblue}\em}

\usepackage{apptools}
\newcommand{\restateref}[1]{\IfAppendix{\hyperref[#1]{$\star$}}{\hyperref[#1*]{$\star$}}}

\graphicspath{{figures/}}

\begin{document}

  \maketitle
  \begin{abstract}
    Train timetables can be represented as \emph{event graphs}, where
    \emph{events} correspond to a train passing through a location at a 
    certain point in time.
    A visual representation of an event graph is important for many
    applications such as dispatching and (the development of) dispatching
    software.
    A common way to represent event graphs are \emph{time-space diagrams}.
    In such a diagram, key locations are visualized on the y-axis and time 
    on the x-axis of a coordinate system. 
    A train's movement is then represented as a connected sequence of line 
    segments in this coordinate system. 
    This visualization allows for an easy detection of infrastructure conflicts
    and safety distance violations.
    However, time-space diagrams are usually used only to depict event graphs
    that are restricted to corridors, where an obvious ordering of the locations 
    exists. 

    In this paper, we consider the visualization of general event
    graphs in time-space diagrams, where the challenge is to find an
    ordering of the locations that produces readable drawings.  We
    argue that this means to minimize the number of \emph{turns},
    i.e., the total number of changes in y-direction.    
    To this end, we establish a connection between this problem and
    \textsc{Maximum Betweenness}.
    Then we develop a preprocessing strategy to reduce the instance
    size. We also propose a parameterized algorithm and integer linear
    programming formulations. We experimentally evaluate the preprocessing 
    strategy and the integer programming formulations on a real-world 
    dataset. Our best algorithm solves every instance in the dataset
    in less than a second.  This suggests that turn-optimal
    time-space diagrams can be computed in real time.
  \end{abstract}

  \section{Introduction}
  Train schedules are subject to constant changes due to interferences
  such as temporary infrastructure malfunctions or congestion
  resulting from
  high traffic volume. As a consequence, train schedules
  must be adjusted in real-time to remedy the disturbances via rerouting
  and other means. In recent years, the computer-assisted execution 
  of this process has gained track. 
  DB InfraGO AG, a subsidiary of Deutsche Bahn AG, is
  developing approaches based on so-called \emph{event graphs}~\cite{event-graphs}
  as an underlying structure that encodes the
  necessary information to (re-)compute a train schedule.
  An event graph models trains running on specific routes on an 
  infrastructure via events.

  \begin{definition}[Event Graph]
    An \emph{event graph} $\mathcal{E}$ is a directed graph.
    Let $V(\mathcal{E})$ denote the vertex set of~$\mathcal{E}$.
    Each vertex~$v$
    of~$\mathcal{E}$, called \emph{event}, is associated with a
    location $\locationOf{v}$, a positive integer $\trainOf{v}$, and a
    point of time $\timeOf{v}$ when the event is scheduled.  For two
    different events~$u$ and $w$, if $\timeOf{u}=\timeOf{w}$, then
    $\trainOf{u} \ne \trainOf{w}$ and
    $\locationOf{u} \ne \locationOf{w}$.  There is an arc $(u,w)$
    in~$\mathcal{E}$ if (i)~$\trainOf{u}=\trainOf{w}$,
    (ii)~$\timeOf{u} < \timeOf{w}$, and (iii)~there is no event $v$ with
    $\trainOf{v}=\trainOf{u}$ and
    $\timeOf{u} < \timeOf{v} < \timeOf{w}$.
  \end{definition}
  For a train~$z$, we call the sequence $v_1,\dots,v_j$ of all events
  with $\trainOf{v_1} = \dots =\trainOf{v_j} = z$ ordered by
  $t(\cdot)$ the \emph{train line} of train $z$.

  For the further automation and for real-time human intervention with
  timetables, it is important that large event graphs can be easily understood
  by humans.

  If the event graph corresponds to trains running on a corridor, i.e., 
  trains running from point $a$ to point $b$ in a linear piece of 
  infrastructure,
  \emph{time-space diagrams} are a common way to represent the event graph.  
  A time-space diagram can be described as a straight-line drawing of the
  event graph with the additional constraint that all vertices that belong
  to the same location lie on the same horizontal line and that the 
  x-coordinate of each vertex is given by its point in time.

  In this paper, we investigate the possibility to use time-space diagrams
  for visualizations for general event graphs, i.e., event graphs that are 
  generally not based on a linear piece of infrastructure.
  We are not aware of any previous work on the visualization of general
  event graphs.
  Formally, a time-space diagram can be defined as follows.

  \begin{definition}[Time-Space Diagram]\label{def:time-space-diagram}
    Let $\mathcal{E}$ be an event graph, let
    $Y=|\locationOf{V(\mathcal{E})}|$, and let
    $y\colon \locationOf{V(\mathcal{E})} \to \{1, 2, \dots, Y\}$ be a
    bijection.  The \emph{time-space diagram} induced by $y$ is the
    straight-line drawing of $\mathcal{E}$ in the plane where event
    $v$ is mapped to the point $(\timeOf{v},y(\locationOf{v}))$.
  \end{definition}

  In a time-space diagram (see \cref{fig:diagram-example}a for two
  examples), we call $y(p)$ the \emph{level} of location~$p$.
  We also refer to a time-space diagram of an event graph as a drawing
  of the event graph.
  \begin{figure}
    \centering
    \includegraphics{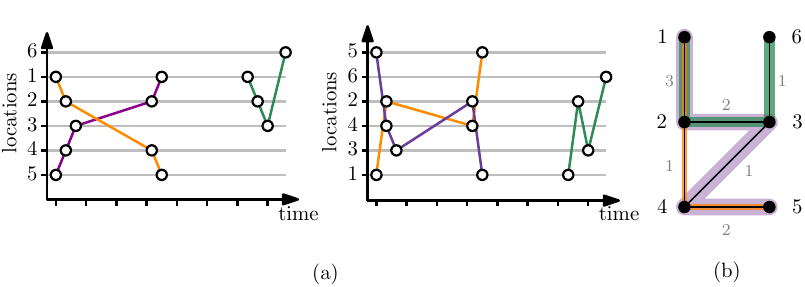}
    \caption{(a) Two different time-space diagrams of the same event
      graph~$\mathcal{E}$ with locations $\{1, \dots, 6\}$.
      (b)~The location graph of~$\mathcal{E}$; the colored paths are
      the train lines, the gray numbers are the weights.}
    \label{fig:diagram-example}
  \end{figure}
  If an event graph is based on a corridor, the corridor induces a
  natural order of the locations: consecutive locations are assigned
  consecutive levels. In this way, train lines in the time-space diagram
  correspond to polylines that go only up- or downwards.
  However, on general event graphs where the trains do not run on a
  linear piece of infrastructure, this intuitive assignment is far from
  trivial, and it is not immediately clear which criteria yield a
  comprehensible drawing.
  \cref{fig:crossings-vs-turns} shows two possible time-space diagrams
  for the same event graph, one that minimizes the number of crossings
  of line segments 
  (a classical objective in graph drawing) and one that minimizes the 
  number of \emph{turns} which we define as follows.

  \begin{figure}
   \centering
    \includegraphics[scale=0.75]{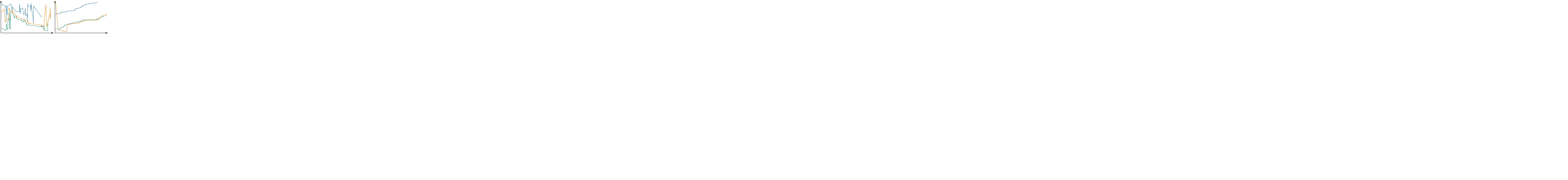}
    \caption{Two time-space diagrams of the same event graph.
    Left: A crossing-minimal drawing with zero crossings (and 71 turns).
    Right: A turn-minimal drawing with one turn (and five crossings).}
   \label{fig:crossings-vs-turns}
  \end{figure}

  Given a drawing $\Gamma$ of an event graph $\mathcal{E}$ and three 
  consecutive events of a train line in $\mathcal{E}$ with pairwise 
  distinct locations $p$, $q$, $r$, we say that there
  is a \emph{turn} in $\Gamma$ if the level of $q$ is smaller/larger than 
  the levels of $p$ and $r$.\footnote{Note that this definition does not
  consider the case where consecutive events have the same location.
  It is easy to see, however, that we can normalize any event graph such that 
  consecutive events always have different locations without changing
  the number of turns in an optimal solution.}

  \cref{fig:crossings-vs-turns} and further experiments suggest that
  minimizing the number of turns leads to time-space diagrams that are 
  significantly better to interpret than drawings that minimize 
  the number of crossings.
  Therefore, we consider the following problem in this paper.
  
  \begin{problem}[\textsc{Turn Minimization}]
    \label{def:tstm}
    Given an event graph~$\mathcal{E}$, find a time-space diagram
    of~$\mathcal{E}$ that minimizes the total number of turns along
    the train lines defined by~$\mathcal{E}$.
  \end{problem}

  \subparagraph*{A connection to \textsc{Maximum Betweenness}.}
  Note that the number of turns in a time-space diagram is determined
  solely by the function~$y$ which represents an ordering of the locations.
  Therefore, \textsc{Turn Minimization} is closely related to the following problem. 
  \begin{problem}[\textsc{Maximum Betweenness} \cite{opatrny-betweenness}]
    \label{def:betweenness}
    Let $S$ be a finite set, and let $R \subseteq S\times S\times S$ be a
    finite set of ordered triplets called \emph{restrictions}. A total
    order $\prec$ \emph{satisfies} a restriction $(a, b, c) \in R$ if  
    either  $a \prec b \prec c$ or $c \prec b \prec a$ holds. Find
    a total order that maximizes the number of satisfied restrictions.
  \end{problem}

  In fact, there is a straightforward translation that transforms optimal
  solutions of \textsc{Turn Minimization} to optimal solutions of \textsc{Maximum Betweenness}
  and vice versa.
  However, note that the objective functions of these problems differ. 
  In \textsc{Turn Minimization} we minimize the number of turns, which corresponds to 
  minimizing the number of \emph{unsatisfied} restrictions in 
  \textsc{Maximum Betweenness}. 

  \textsc{Maximum Betweenness} and other slightly modified variants have been studied 
  extensively \cite{max-betweenness-max-snp,max-betweenness-genetic-algo,
  force-based-max-betweenness,mip-max-betweenness,consecutive-ones-and-betweenness,
  betweenness-application} mainly motivated by 
  applications in biology.
  In particular, Opatrny showed that \textsc{Maximum Betweenness}
  is \NP-hard~\cite{opatrny-betweenness}.
  Further, the problem admits 1/2-approximation 
  algorithms~\cite{geometric-betweenness,linear-time-approx-max-betweennness}, 
  but for any $\varepsilon>0$ it is \NP-hard to compute a 
  $(1/2+\varepsilon)$-approximation~\cite{inapprox}.
  
  Several exact algorithms have been proposed.
  There is an intuitive integer linear program formulation that uses
  linear orderings.
  This formulation has been used in various variants and algorithms as
  a baseline before \cite{consecutive-ones-and-betweenness,betweenness-application}.
  We state this formulation in \cref{sec:ilps} since we also use this formulation
  as a starting point.
  Note that this formulation requires $\mathcal{O}(|S|^3)$ many 
  constraints due the transitivity constraints for the linear ordering and
  the fact that there can be $\mathcal{O}(|S|^3)$ many restrictions. 
  This may result in long computation times.
  Since applications in biology often deal with a large number of restrictions,
  improvements of this formulation have been made via cutting-plane
  approaches that target the transitivity and the restriction constraints 
  using the trivial lifting technique\footnote{Here: complete linear 
  descriptions of smaller instances $(S, R)$ are used to generate valid 
  inequalities for larger instances $(S', R')$ where $S \subseteq S'$, $R \subseteq R'$.}
  \cite{consecutive-ones-and-betweenness,betweenness-application}.
  To overcome the issue of a cubic number of transitivity constraints, 
  a mixed-integer linear program has been 
  proposed \cite{mip-max-betweenness} that circumvents the large 
  number of constraints entirely by modelling a linear ordering with continuous
  variables that are required to be distinct.
  As a consequence, this formulation runs faster than the intuitive formulation.
  However, this program requires a user-specified parameter that influences
  the runtime significantly and is not obvious to choose.

  Note that due to the relationship between \textsc{Turn Minimization} and 
  \textsc{Maximum Betweenness}, exact algorithmic approaches for one of the 
  problems can be directly transferred to the other, while approximation
  guarantees cannot.
  However, most of the algorithmic approaches for \textsc{Maximum Betweenness}
  are optimized for instances where the ratio between~$|R|$ and~$|S|$ is large, 
  while this ratio is only moderate in our setting.
  As a result, in our setting, the transitivity constraints become 
  the bottleneck as opposed to the constraints that model the restrictions in~$R$.
  Another notable difference is that in \textsc{Turn Minimization} we have access to 
  additional information on the relations between restrictions (the train lines). 
  We strive to leverage these differences to develop new approaches. 
  In particular, we consider the case that instances admit
  drawings with a small number of turns (as otherwise a drawing would not 
  be comprehensible). Further, we consider the case that 
  the underlying infrastructure of the event graph is sparse (as this is
  common in train infrastructure).

  \subparagraph*{Contribution.} 
  First, we consider \textsc{Turn Minimization} from a
  (parameterized) complexity theoretic perspective; see \cref{sec:problem-complexity}.
  In particular, we show that it is \NP-hard to compute an
  $\alpha$-approximation for any constant $\alpha \geq 1$ and that 
  the problem is \paraNP-hard when parameterized by the number of turns.
  The problem is also \paraNP-hard when parameterized 
  by the vertex cover number of the location graph, a graph that represents 
  the infrastructure.
  Second, we propose a preprocessing strategy that reduces a given
  event graph~$\mathcal{E}$ into a smaller event graph~$\mathcal{E}'$ that
  admits drawings with the same number of turns; see \cref{sec:reduction-rule}.
  Third, we refine the intuitive integer linear program in two
  different ways; see \cref{sec:ilps}. 
  The first refinement is a simple cutting-plane approach that iteratively adds
  transitivity constraints until a valid (optimal) solution is found.
  The second refinement uses a tree decomposition to find a light-weight 
  formulation of the problem.
  We conclude with an experimental analysis and future work in 
  \cref{sec:experiments,sec:conclusion}.
  While the preprocessing strategy cannot be easily translated to the
  \textsc{Maximum Betweenness} problem, our remaining results carry over.

  \section{Preliminaries}\label{sec:preliminaries}
  Let $S$ be a finite set, and $k$ be a positive integer. Let $[S]^k$ denote
  the set $\{X \colon X \subseteq S, |X| = k\}$ of $k$-element subsets of $S$.
  We call $(A, B)$ with $A\cup B = V(G)$ a \emph{separation} of a graph~$G$ 
  if, on every $a$--$b$ path with $a \in A$ and $b \in B$, there is at least
  one vertex in $A\cap B$. We call $A\cap B$ a \emph{separator}.
  If a separator is a single vertex, we call this vertex a
  \emph{cut vertex}. If a separator consists of two vertices,
  then we call the two vertices a \emph{separating pair}.
  We say that a connected graph is \emph{biconnected} 
  if it does not contain a cut vertex. 
  Similarly, we say that a biconnected graph is \emph{triconnected} 
  if it does not contain a separating pair.
  A partition of $(A, B)$ of the vertex set of a graph is a \emph{cut}.
  The \emph{cut set} of $(A, B)$ is the set of edges with one vertex in~$A$ 
  and one vertex in~$B$. The size of a cut is the size of the cut set.

  Problems that can be solved in time $f(k)\cdot n^c$, where $f$ is a
  computable function, $c > 0$ is a constant, $n$ the input size, and $k$
  is a parameter, are known as \emph{fixed-parameter tractable}
  (parameterized by the parameter $k$). The complexity class \textsf{FPT}
  contains precisely all such fixed parameter tractable problems with the
  respective parameter $k$. 
  If a problem remains \NP-hard even on instances, where the parameter $k$
  is bounded, we say that  the problem is \paraNP-hard when parameterized
  by $k$.
  The study of \FPT{} algorithms is particularly
  motivated by scenarios where certain instances have properties, described
  by the parameter, that are small or constant, making \FPT{} algorithms
  efficient for these instances.
  
  We define two auxiliary graphs that capture the connections 
  between locations in~$\mathcal{E}$, which we use in our algorithms.

  \begin{definition}[Location Graph]
    Let $\mathcal{E}$ be an event graph.  The \emph{location graph}
    $\mathcal{L}$ of~$\mathcal{E}$ is an undirected weighted graph
    whose vertices are the locations of $\mathcal{E}$.
    For two locations $p \ne q$, the weight $w(\{p, q\})$ of the 
    edge~$\{p, q\}$ in $\mathcal{L}$ corresponds to the number of
    arcs $(u,v)$ or $(v,u)$ in the event graph $\mathcal{E}$
    such that $\locationOf{u} = p$ and $\locationOf{v} = q$ and
    $\trainOf{u} = \trainOf{v}$.
    If $w(\{p, q\})=0$, then~$p$ and $q$ are not adjacent in~$\mathcal{L}$.
  \end{definition}
  See \cref{fig:diagram-example}b for an example of a location graph.
  Note that the train line of a train $z$ in the event graph 
  corresponds to a walk (a not necessarily simple path) in the location
  graph. Slightly abusing notation, we also call this walk in
  the location graph a \emph{train line} of~$z$. 

  \begin{definition}[Augmented Location Graph]\label{def:augmented-location-graph}
    Let $\mathcal{E}$ be an event graph. 
    The \emph{augmented location graph} $\mathcal{L}'$ of $\mathcal{E}$ is
    the supergraph of the location graph $\mathcal{L}$ of $\mathcal{E}$
    that additionally contains, for each triplet
    $(v, v', v'')$ of locations that are consecutive
    along a train line in~$\mathcal{E}$, the edge
    $\{\locationOf{v}, \locationOf{v''}\}$.
  \end{definition}
  The augmented graph $\mathcal{L}'$ has the crucial property that
  every three such consecutive events, whose locations can potentially
  cause a turn, induce a triangle in~$\mathcal{L}'$.

  \subparagraph{Tree decompositions.}

  Intuitively, a tree decomposition is a decomposition of a 
  graph~$G$ into a tree $T$ which gives structural information 
  about the separability of $G$. 
  The treewidth of a graph $G$ is a measure that captures how 
  similar a graph $G$ is to a tree. 
  For instance, every tree has treewidth~1, the graph of a 
  $(k \times k)$-grid has treewidth~$k$, and the treewidth 
  of the complete graph~$K_n$ is $n-1$.  More formally, a \emph{tree
    decomposition} $\mathcal{T} = (T, \{X_t\}_{t \in V(T)})$ of~$G$
  consists of a tree~$T$ and, for each node~$t$ of~$T$, of a
  subset~$X_t$ of $V(G)$ called \emph{bag} such that (see
  \cref{fig:tree-decomp-example} for an example):
  \begin{enumerate}[(T1)]
  \item the union of all bags is $V(G)$,

  \item\label{prop:tree-decomp-t2} for every edge $\{u, v\}$ of~$G$,
    the tree~$T$ contains a node~$t$ such that
    $\{u, v\} \subseteq X_t$, and

  \item\label{prop:tree-decomp-t3} for every vertex~$v$ of~$G$, the
    nodes whose bags contain $v$ induce a connected subgraph of~$T$.
  \end{enumerate}
  \begin{figure}
    \centering
    \includegraphics[page=1]{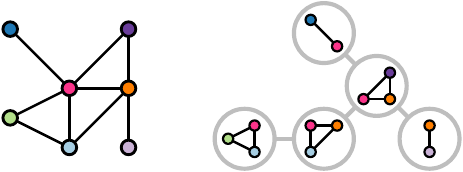}
    \caption{Example of a tree decomposition (right) of the graph on the left.
      Each bag of the tree decomposition is depicted as the graph it induces.}
    \label{fig:tree-decomp-example}
  \end{figure}
  The \emph{width} of a tree decomposition is defined as 
  $\max\{|X_t|\colon t \in V(T)\}-1$; for example, the tree
  decomposition in \cref{fig:tree-decomp-example} has width~2.
  The \emph{treewidth} of a graph~$G$,
  $\operatorname{tw}(G)$, is the smallest value such that 
  $G$ admits a tree decomposition of this width. 
  Any tree decomposition $(T, \{X_t\}_{t \in V(T)})$ has the
  following properties.
  \begin{enumerate}[(P1)]
  \item\label{prop:tree-decomp-separator} For every edge $\{a, b\}$ of
    $T$, the graph $T - \{a,b\}$ has two connected components~$T_a$
    and~$T_b$, where $T_a$ contains~$a$ and $T_b$ contains~$b$.  They
    induce a separation
    $(A, B) = (\bigcup_{t\in V(T_a)} X_t, \bigcup_{t\in V(T_b)}X_t)$
    of $G$ with separator $X_a \cap X_b$. In particular,
    $A \cap B = X_a \cap X_b$, and $G$ does not contain any edge
    between a vertex in $A\setminus X_a$ and a vertex in
    $B\setminus X_b$.

  \item\label{prop:tree-decomp-clique} For every clique $K$ of $G$,
    the tree~$T$ contains a node~$t$ such that $V(K) \subseteq X_t$.
  \end{enumerate}

  \section{Problem Complexity}
  \label{sec:problem-complexity}
  
  We study the approximability of \textsc{Turn Minimization} and consider the tractability
  of \textsc{Turn Minimization} with respect to parameters that we expect to be small
  in our instances. We obtain negative results for the approximability and for most of
  the considered parameters, but we propose an \FPT{} algorithm
  parameterized by the treewidth of the augmented location graph.
  
  While \NP-hardness of \textsc{Turn Minimization} is easy to see due to the one-to-one 
  correspondence (of problem instances and optimal solutions)
  to \textsc{Maximum Betweenness}, 
  approximability results for \textsc{Maximum Betweenness} do not carry over
  because of the different objectives.
  In fact, by close inspection of the natural transformation between instances
  of the two problems (that we give in the proof of \cref{thm:problem-complexity}),
  we are able to deduce that there is neither a multiplicative nor an additive 
  constant factor approximation algorithm for \textsc{Turn Minimization}, unless $\Polytime=\NP$.
  Furthermore, we can derive from the reduction that, unless $\Polytime=\NP$, 
  there is no efficient algorithm even for instances where the number of turns is bounded. 
  By a reduction from the decision version of \textsc{MaxCut}, we can also show that
  the problem is \NP-hard when parametrized by the vertex cover number.
  \textsc{MaxCut} asks whether there is a cut of size at least $k$ in a given graph.

  \begin{restatable}[\restateref{thm:problem-complexity}]{theorem}{problemcomplexity}
    \label{thm:problem-complexity}
    \textsc{Turn Minimization}
    \begin{enumerate}[(i)] 
    \item\label{para-np-turns} is  \paraNP-hard with respect to the
      natural parameter, the number of turns,
    \item\label{para-np-vc} is  \paraNP-hard with respect to the
      vertex cover number of the location graph,
    \item\label{inapprox} does not admit polynomial-time multiplicative or
      additive approximation algorithms unless $\Polytime=\NP$.
    \end{enumerate}
  \end{restatable}\medskip

  Note that \eqref{para-np-vc} also implies that \textsc{Turn Minimization} is \paraNP-hard
  if parameterized by the treewidth of the location graph $\mathcal{L}$ since
  the treewidth of a graph is bounded by the vertex cover number of the graph.
  Therefore, there is no  fixed-parameter tractable algorithm
  with respect to the treewidth of $\mathcal{L}$, unless $\Polytime = \NP$. 
  However, we obtain the following result.

  \begin{restatable}[\restateref{thm:fpt-tw}]{theorem}{paramalgo}
    \label{thm:fpt-tw}
    Let $\mathcal{E}$ be an event graph, and let $\mathcal{L}'$ be its
    augmented location graph.  Computing a turn-optimal time-space
    diagram of $\mathcal{E}$ is fixed-parameter tractable
    with respect to the treewidth of $\mathcal{L}'$.
  \end{restatable}
  \begin{proof}[Proof sketch.]
    For every triplet of consecutive events, the corresponding
    locations form a triangle in~$\mathcal{L}'$.  Hence, due to
    Property~(P\ref{prop:tree-decomp-clique}), in any tree
    decomposition of~$\mathcal{L}'$, there is a bag that contains the
    three locations.  Thus, every potential turn occurs in at least
    one bag, and it suffices to run a standard dynamic program over a
    (nice) tree decomposition of~$\mathcal{L}'$.
  \end{proof}

  Note that this result might not be practical since the treewidth of the
  augmented location graph might be considerably
  larger than the treewidth of the location graph.

  \section{An Exact Reduction Rule}\label{sec:reduction-rule}
  
  In this section we describe how to reduce the event graph based
  on the identification and contraction of simple substructures in the
  location graph.
  Consider the location graph $\mathcal{L}$ of an event graph~$\mathcal{E}$. 
  We call a vertex $p$ of~$\mathcal{L}$ a \emph{terminal} if a train starts 
  or ends at~$p$. We say that a path in $\mathcal{L}$ is a \emph{chain} 
  if each of its vertices has degree exactly~2 in~$\mathcal{L}$ and the 
  path cannot be extended without violating this property. 
  If a chain contains no terminals, and no train line restricted to this
  chain induces a cycle, then there is always a turn-minimal drawing of
  $\mathcal{E}$ 
  that contains no turn along the chain. 
  This is due to the fact that any turn on the chain
  can be moved to a non-chain vertex adjacent to one of
  the chain endpoints.  We now generalize this intuition, assuming
  that $\mathcal{L}$ is not triconnected.

  Let $\{s,t\}$ be a separating pair of~$\mathcal{L}$, let $C$ be a
  connected component of $\mathcal{L}\setminus\{s,t\}$, and let 
  $C'=\mathcal{L}[V(C) \cup \{s,t\}]$.  
  We call $C$ a \emph{transit component} if (i) $C$ does not 
  contain any terminal, and (ii) the trains passing through~$C'$ 
  via $s$ also pass through~$t$ (before possibly passing through $s$ again).
  A transit component~$C$ is \emph{contractible} if, for each path
  associated to a train line in~$C$, we can assign a direction such
  that the resulting directed graph is acyclic.

  \begin{reductionrule}[Transit Component Contraction]
    \label{rule:transit-components}
    Let $\mathcal{E}$ be an event graph, and let $\mathcal{L}$ be the
    location graph of~$\mathcal{E}$.  If $\mathcal{L}$ is not
    triconnected, then let $\{s,t\}$ be a separating pair
    of~$\mathcal{L}$ and let~$C$ be a contractible transit component
    of~$\mathcal{L} \setminus \{s,t\}$.
    For each train~$z$ that traverses~$C$, replace in~$\mathcal{E}$ the
    part of the train line of $z$ between the events that correspond
    to~$s$ and~$t$ by the arc (directed according to time)
    that connects the two events.
  \end{reductionrule}

  \begin{theorem}
    \label{lemma:soundness-transit-component}
    Let $\mathcal{E}$ be an event graph.  If $\mathcal{E}'$ is an
    event graph that results from applying
    \cref{rule:transit-components} to~$\mathcal{E}$, then a
    turn-optimal drawing of~$\mathcal{E}$ and a turn-optimal drawing
    of~$\mathcal{E}'$ have the same number of turns.
  \end{theorem}
  \begin{proof}
    Let $k$ be the number of turns in a turn-optimal drawing $\Gamma$
    of an event graph $\mathcal{E}$ and let $k'$ be the number of turns
    in a turn-optimal drawing $\Gamma'$ of an event graph $\mathcal{E}$
    after \cref{rule:transit-components} was applied on the contractible 
    transit component~$C$. Further, for the sake of simplicity, assume 
    that any train traversing~$C$ traverses~$C$ only once. 
    If a train $z$ traverses~$C$ multiple times, the same arguments apply
    for each connected component of the train line of $z$ going through~$C$.
    Note that this can indeed happen (if the event graph represents a
    schedule that contains a train that moves through~$C$ periodically).

    Without increasing the number of turns, we now transform $\Gamma'$
    into a drawing of $\mathcal{E}$ that contains~$C$, as follows.
    Let $\{s,t\}$ be the pair that separates $C$ from
    $\mathcal{L} \setminus C$.  Let $y'(s)$ and $y'(t)$ be the levels
    of~$s$ and~$t$ in~$\Gamma'$, respectively.  Without loss of
    generality, assume that $y'(s) < y'(t)$.
    Since $C$ is contractible, $C$ admits an acyclic (topological)
    ordering of the locations in $C$ such that the train line of every 
    train traversing $C$ is directed from $s$ to $t$.
    Let $\prec_C$ be such an ordering and let $z$ be a train traversing
    $C$, where $z' = \langle v_1, \dots, v_j\rangle$ is the component of
    the train line of $z$ that traverses $C$ such that 
    $\locationOf{v_1} = s$ and $\locationOf{v_j}=t$. Since $\prec_C$ is a
    valid acyclic topological ordering, it holds that 
    $\locationOf{v_i} \prec_C \locationOf{v_{i+1}}$ for each $1 \leq i < j$. 
    Thus, by extending $y'$ so that each vertex $p$ in $C$ is assigned a 
    level $y'(p) \in ]y'(s), y'(t)[$, with $p, q \in V(C)$ and 
    $y'(p) < y'(q)$ if and only if $p \prec_C q$, no additional turns are 
    introduced in the transformed drawing. As a result, we obtain a
    drawing of $\mathcal{E}$ that contains $C$ and has the same number of
    turns as $\Gamma'$.  Hence $k \leq k'$.

    Conversely, we transform $\Gamma$ into a drawing with at most $k$
    turns where the component~$C$ is contracted.  Let $\{s,t\}$ be the
    pair that separates $C$ from $\mathcal{L} \setminus C$.  Let
    $y(s)$ and $y(t)$ be the levels of $s$ and $t$ in~$\Gamma$.
    Without loss of generality, $y(s) < y(t)$.
    First, assume that for each $p \in V(C)$ it holds that
    $y(s) < y(p) < y(t)$.  Then, contracting~$C$ into a single
    edge transforms~$\Gamma$ into a drawing of the reduced instance with
    at most $k$ turns. 
    Now, let $L \subseteq V(C)$ be the set of vertices below~$s$, and
    let $U \subseteq V(C)$ be the set of vertices above~$t$.
    We describe only how to handle $U$ since $L$ can be handled analogously.
    Let $\Delta$ be the number of train lines with at least one vertex
    in $U$. We reorder the levels of the vertices in $U$ 
    according to a topological ordering of $C$ restricted to $U$ and
    move all vertices in $U$ such that their levels
    are between $y(t)$ and
    $Y = \max \{ y(p') \colon y(p') < y(t), p'\in V(C) \}$. 
    Each train line with at least one vertex in $U$ corresponds to at
    least one turn in $U$, namely a turn at the vertex of a train line
    with the largest level.
    Therefore, moving and reordering $U$ removes at least $\Delta$ turns.
    Also, the movement and reordering of $U$ results in at most $\Delta$
    turns more at $t$ since the only vertices that were moved are
    vertices in $U$. 
    After moving and reordering $U$ and $L$, we are in the first case
    and can hence contract~$C$.
    Summing up, we can transform a turn-optimal drawing $\Gamma$ of
    an event graph $\mathcal{E}$ into a drawing with a contracted
    component $C$ without changing the number of turns, implying
    that $k' \leq k$.
    
    We conclude that \cref{rule:transit-components} is sound. 
  \end{proof}

  Note that we can apply \cref{rule:transit-components} exhaustively
  in cubic time (in the number of vertices and edges of~$\mathcal{L}$ 
  and~$\mathcal{E}$): we
  iterate over all (up to $\mathcal{O}(|V(\mathcal{L})|^2)$ many)
  separating pairs and contract each connected component with respect
  to the current pair, if possible.  Below we show that, using two
  data structures, we can speed up the application of
  \cref{rule:transit-components} considerably.

  \subparagraph*{Block-cut trees.}

  A \emph{block-cut tree} represents a decomposition of a graph into
  maximal biconnected components (called \emph{blocks}) and cut
  vertices.  Given a graph~$G$, let $\mathcal{C}$ be the set of cut
  vertices of~$G$, and let $\mathcal{B}$ be the set of blocks of~$G$.
  Note that two blocks share at most one vertex with each other,
  namely a cut vertex.  The block-cut tree $\mathcal{T}_\mathrm{bc}$ of~$G$ (see
  \cref{fig:bc-tree-example} for an example) has a node for each
  element of $\mathcal{B} \cup \mathcal{C}$ and an edge between
  $b \in \mathcal{B}$ and $c \in \mathcal{C}$ if and only if $c$ is
  contained in the component represented by $b$.
  Note that the leaves of a block-cut tree are $\mathcal{B}$-nodes.
  For example, the block-cut tree of a biconnected graph is a single node.
  The block-cut tree of an $n$-vertex path is itself a path,
  with $2n-3$ nodes.

  \begin{figure}
    \centering
    \includegraphics[page=2]{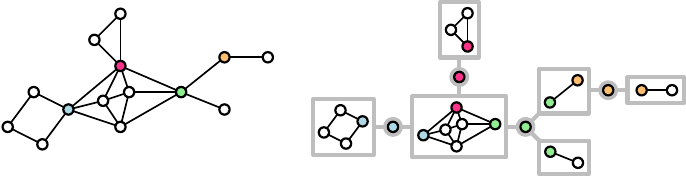}
    \caption{The block-cut tree (right) of the graph (left). The cut vertices
    are colored and are represented as circles in the block-cut tree. The maximal
    biconnected components are represented as rectangles.}
    \label{fig:bc-tree-example}
  \end{figure}

  \subparagraph*{SPQR-trees.}
  If a graph $G$ is biconnected, an \emph{SPQR-tree} represents the 
  decomposition of $G$ into its triconnected components via separating pairs, 
  where S (series), P (parallel), Q (a single edge), and R (remaining or rigid) 
  stand for the different node types of the tree that represent 
  how the triconnected components compose $G$. 
  An SPQR-tree represents all planar embeddings of a graph.
  Therefore, SPQR-trees are widely used in graph drawing and beyond 
  \cite{mutzel-spqr-applications}.
  SPQR-trees are defined in several ways in the literature.
  Here, we recall the definition of Gutwenger and
  Mutzel~\cite{spqr-tree-linear-time}, which is based on an earlier
  definition of Di Battista and Tamassia~\cite{spqr-on-line-planarity}.
  
  Let $G$ be a multi-graph.
  A \emph{split pair} is either a separating pair or a pair of adjacent vertices.
  A \emph{split component} of a split pair $\{u, v\}$ is either an edge $\{u, v\}$
  or a maximal subgraph $C$ (containing $\{u, v\}$) of $G$ such that $\{u, v\}$ 
  is not a split pair of $C$.
  Let $\{s, t\}$ be a split pair of $G$.
  A \emph{maximal split pair} $\{u, v\}$ of $G$ with respect to 
  $\{s, t\}$ is a pair such that, for any other split pair $\{u', v'\}$, 
  the vertices $u$, $v$, $s$, and $t$ are in the same split component.
  Let $e=\{s, t\}$ be an edge of $G$ called \emph{reference edge}.
  The SPQR-tree $\mathcal{T}_\mathrm{spqr}$ of $G$ with respect to~$e$ is a rooted tree whose
  nodes are of four types: S, P, Q, and R. Each node $\mu$ of $\mathcal{T}_\mathrm{spqr}$ 
  has an associated biconnected multi-graph called the \emph{skeleton} of $\mu$.
  Every vertex in a skeleton corresponds to a vertex in $G$ and every edge $\{u, v\}$
  in a skeleton corresponds to a child of $\mu$ that represents a split component
  of $G$ that is separated by $\{u, v\}$. 
  The tree $\mathcal{T}_\mathrm{spqr}$ is recursively defined as follows (see
  \cref{fig:spqr-tree-example} for an example):
  \begin{figure}
    \centering
    \includegraphics[page=3]{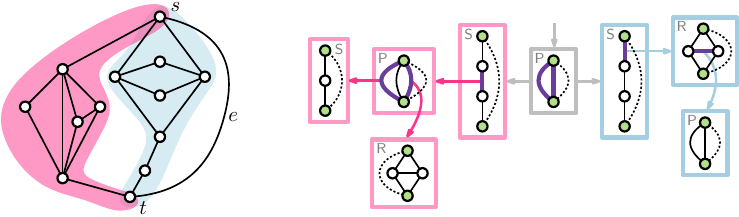}
    \caption{The SPQR-tree (right) of the graph on the left,
      with respect to the edge~$e$.  The
      nodes of the tree are the rectangles.  Each rectangle contains
      the skeleton of the corresponding node.  The dashed edge in the
      representation of a node~$\mu$ is the virtual edge of the parent
      of~$\mu$.  The thick purple edges in the skeletons represent
      child components.  The solid black edges are the real edges of
      the graph.  Q-nodes are omitted for simplicity.  In each
      rectangle, the two green vertices are the poles of the
      corresponding skeleton. The grey P-node is the unique child
      of the root (a Q-node) which represents the reference edge $e$.}
    \label{fig:spqr-tree-example}
  \end{figure}
  \begin{description}
    \item[Trivial Case:] If $G$ consists of exactly two parallel edges between
    $s$ and $t$, then $\mathcal{T}_\mathrm{spqr}$ consists of a single Q-node whose skeleton is
    $G$ itself.
    \item[Parallel Case:] If the split pair $\{s, t\}$ has at least three split
    components $G_1, \dots, G_k$, the root of $\mathcal{T}_\mathrm{spqr}$ is a P-node $\mu$
    whose skeleton consists of $k$ parallel edges $e_1, \dots, e_k$ between
    $s$ and~$t$.
    \item[Series Case:] If the split pair $\{s, t\}$ has exactly two 
    split components, one of them is $e$, and the other is denoted with $G'$.
    If $G'$ has cut vertices $c_1, \dots, c_{k-1}$ ($k \geq 2$) that partition
    $G$ into its blocks $G_1, \dots, G_k$, in this order from $s$ to $t$,
    the root of $\mathcal{T}_\mathrm{spqr}$ is an S-node $\mu$ whose skeleton is the cycle
    $e_0, e_1, \dots, e_k$, where $e_0 = e$, $c_0 = s$, $c_k = t$, and 
    $e_i=(c_{i-1}, c_i)$ for $i = 1, \dots, k$.
    \item[Rigid Case:] If none of the above cases applies, 
    let $\{s_1, t_1\}, \dots, \{s_k, t_k\}$ be the maximal split pairs of $G$ 
    with respect to $\{s, t\}$ and let $G_i$ ($i=1, \dots, k$) be the union
    of all the split components of $\{s_i, t_i\}$ but the one containing $e$.
    The root of $\mathcal{T}_\mathrm{spqr}$ is an R-node whose skeleton is obtained from $G$
    by replacing each subgraph $G_i$ with the edge $e_i=\{s_i, t_i\}$.
  \end{description}

  Except for the trivial case, every node~$\mu$ of~$\mathcal{T}_\mathrm{spqr}$ 
  has children $\mu_1, \dots, \mu_k$ such that
  $\mu_i$ is the root of the SPQR-tree of $G_i\cup e_i$ with respect to $e_i$.
  The reference edge $e$ is represented by a Q-node which is the root of 
  $\mathcal{T}_\mathrm{spqr}$.
  Each edge~$e_i$ in the skeleton of $\mu$ is associated with the child $\mu_i$
  of $\mu$. This edge is also present in $\mu_i$ and is called \emph{virtual edge}
  in~$\mu_i$. The endpoints of edge $e_i$ are called \emph{poles} of the node~$\mu_i$.
  The \emph{pertinent graph} $G_\mu$ of $\mu$ is the subgraph of $G$ that 
  corresponds to the real edges (Q-nodes) in the subtree rooted at~$\mu$.

  \begin{theorem}
    If $\mathcal{L}$ is the location graph of an event graph
    $\mathcal{E}$, then \cref{rule:transit-components} can be
    applied exhaustively in time linear in the number of vertices
    and edges of~$\mathcal{L}$ and~$\mathcal{E}$.
  \end{theorem}
  \begin{proof}
    Let $\mathcal{L}$ be the connected location graph of an event graph $\mathcal{E}$. 
    If the location graph is not connected
    we can apply the algorithm on every connected component of $\mathcal{L}$.
    We consider the case where $\mathcal{L}$ is not biconnected as this case 
    also handles the biconnected subgraphs of $\mathcal{L}$ and can 
    therefore easily be adapted for the case where $\mathcal{L}$ is biconnected.
    We start by decomposing $\mathcal{L}$ into a block-cut tree~$\mathcal{T}_\mathrm{bc}$
    with vertex set $\mathcal{B} \cup \mathcal{C}$.
    For every block $b$ in $\mathcal{B}$ we do the following.
    We construct an SPQR-tree~$\mathcal{T}_{\mathrm{spqr}}^b$ of the biconnected
    component corresponding to the block $b$ and temporarily mark all cut 
    vertices of $\mathcal{C}$ contained in $b$ as terminal in $b$ 
    so that no cut vertex can be contracted while we process 
    $\mathcal{T}_{\mathrm{spqr}}^b$.
    Let $\mu_r$ be the root of $\mathcal{T}_{\mathrm{spqr}}^b$ 
    and let~$\mathcal{E}_\mu$ be the event
    graph that corresponds to the pertinent graph $\mathcal{L}_\mu$ for a
    node $\mu$ in $\mathcal{T}_{\mathrm{spqr}}^b$. 
    We say a train \emph{loops} at vertex $s$ for some split pair $\{s, t\}$, 
    if there is a train whose train line contains the subsequence 
    $\langle t, s, t\rangle$.
    Intuitively, we traverse $\mathcal{T}_{\mathrm{spqr}}^b$ bottom-up, and 
    mark nodes in~$\mathcal{T}_{\mathrm{spqr}}^b$ (and the corresponding edge in their parent) 
    as contractible or non-contractible and modify 
    $\mathcal{E}_\mu$ using \cref{rule:transit-components}.
    Due to the correspondence between vertices in a skeleton and vertices 
    in~$\mathcal{L}$ we use ``vertex in a skeleton'' and 
    ``vertex in~$\mathcal{L}$'' interchangeably.
    We do the following depending on the type of $\mu \neq \mu_r$. 
    \begin{description}
      \item[Q-node:]
      We mark the edge corresponding to $\mu$ in the skeleton of the parent 
      of $\mu$ as contractible (a contraction of the real edge corresponding
      to $\mu$ does not result in a different graph). If there is a train
      that loops at one of the poles of $\mu$, we mark this pole as terminal.

      \item[S-node:] Let $\langle c_1, \dots, c_{k-1}\rangle$ be the path 
      that corresponds to the skeleton of $\mu$ without the virtual edge of
      its parent. 
      Every maximal subpath of contractible edges that does not
      contain any terminal $c_i$ 
      is therefore a contractible transit component. Thus, we can
      apply \cref{rule:transit-components} on the graph that is induced by
      this subpath. 
      If the entire path $\langle c_1, \dots, c_{k-1}\rangle$ is contractible, 
      then we mark the edge corresponding to $\mu$ in the skeleton of $\mu$'s
      parent as contractible, otherwise we mark it as non-contractible.
      In the contractible case, we check if there is a train that
      loops at the poles of $\mu$. If this is the case, we mark
      the corresponding pole(s) as terminal.

      \item[P-node:] Let $e_1, \dots, e_k$ be the edges between the poles
      of $\mu$ without the virtual edge corresponding to the 
      parent of $\mu$. We consider every contractible edge among
      $e_1, \dots, e_k$ as a single transit component $C$ and apply
      \cref{rule:transit-components}. Note that $C$ is a contractible
      transit component since the union of parallel contractible 
      transit components is again a contractible transit component.
      If every edge $e_1, \dots, e_k$ is contractible, we mark the
      edge corresponding to $\mu$ in the skeleton of $\mu$'s parent
      as contractible. Further, we check if there is a train that
      loops at the poles of $\mu$. If this is the case, we mark
      the corresponding pole(s) as terminal.
      
      \item[R-node:] Let $C$ be the skeleton of $\mu$ without the
      virtual edge of its parent.
      If every edge in~$C$ is contractible, then we test if~$C$ is
      a contractible transit component with respect to the poles of~$\mu$.
      If this is the case, we apply \cref{rule:transit-components}
      on~$C$ and mark the edge corresponding to $\mu$ in the skeleton
      of $\mu$'s parent as contractible. Again, if there is a train
      that loops at one of the poles of $\mu$, we mark the
      corresponding pole(s) as terminal.
    \end{description}
    To process the root $\mu_r$ of $\mathcal{T}_{\mathrm{spqr}}^b$, 
    we do the following depending on the type of the single child 
    $\mu_c$ of $\mu_r$.
    If $\mu_c$ is an S-node, we check if the edge corresponding to $\mu_c$ 
    is marked as contractible.
    If this is the case, we mark $b$ as contractible.
    Otherwise, we consider the maximal subpath of the skeleton of
    $\mu_c$ that now contains the edge corresponding to $\mu_r$ and
    apply \cref{rule:transit-components}, if possible.
    If $\mu_c$ is a P- or R-node, we mark $b$ as contractible if 
    $\mu_c$ is also contractible.
    To complete the processing of $b$, we finally check whether 
    there is a train that loops at one of the poles of $\mu_r$, 
    if this is the case we mark $b$ as non-contractible.
    Additionally, if this pole is also a cut vertex, we mark it as terminal.
    It remains to test for one special case. If for every skeleton in
    $\mathcal{T}_{\mathrm{spqr}}^b$ every edge is marked as contractible, except for a 
    single edge $e'$ in an R-node, test if $\mathcal{L}\setminus C_{e'}$ is
    a contractible transit component where $C_{e'}$ is the split component
    of $e'$. If this is the case apply \cref{rule:transit-components}.
    
    After we have completed every block $b$ in $\mathcal{B}$, we mark every
    node $c$ in $\mathcal{C}$ as contractible if $c$ is not a terminal.
    Finally, we proceed similarly to the S-node previously.
    For every chain in~$\mathcal{T}_\mathrm{bc}$ that contains only contractible 
    vertices, we apply \cref{rule:transit-components}.

    A block-cut tree and an SPQR-tree can be computed in linear time each
    \cite{bc-trees,spqr-tree-linear-time}. It is easy to see that
    the total size of all skeletons is linear in the size of the
    location graph.
    Since each node is processed in time linear in
    the size of its skeleton and of the corresponding event graph, 
    the entire algorithm takes time linear in the sizes of~$\mathcal{L}$
    and~$\mathcal{E}$.
  \end{proof}

  \section{Exact Integer Linear Programming Approaches}\label{sec:ilps}  
  To state an ILP for \textsc{Turn Minimization}, we transform a given event 
  graph~$\mathcal{E}$ and its location graph~$\mathcal{L}$ into an equivalent 
  \textsc{Maximum Betweenness} instance $(S, R)$, where $S = V(\mathcal{L})$ and~$R$ 
  contains all triplets of locations of consecutive events in all train lines 
  of~$\mathcal{E}$. 
  However, we state the ILP in the context of \textsc{Turn Minimization} and minimize
  the number of violated constraints. We start with the intuitive integer linear
  program.

  We assume that the set of restrictions~$R$ is ordered arbitrarily,
  and we denote the $i$-th element in $R$ by $(p, q, r)_i$.
  Further, let $U = \{(p, q, r) \colon \{p, r\}\in[S]^2, 
  q\in S, q\neq p, q\neq r\}$.
  For each pair of elements $p, q \in S$ with $p\neq q$, 
  let $x_{pq} \in \{0, 1\}$ be a binary decision variable, 
  where $x_{pq}=1$ means that $p \prec q$.
  We require $x_{pq} = 1 - x_{qp}$ to ensure asymmetry. Further,
  we model the transitivity constraints of an ordering for 
  $(p, q, r) \in U$ using the constraint
  \begin{align*}
    x_{pr} \geq x_{pq} + x_{qr} - 1,
  \end{align*}
  i.e., the constraint ensures that if $p \prec q$ and $q \prec r$, 
  then $p \prec r$ must hold as well.
  
  It remains to count the number of restrictions $(p, q, r)_i \in R$ 
  that are violated. 
  For this purpose, we introduce a binary variable $b_i$ for each
  restriction $(p, q, r)_i \in R$. 
  The intended meaning of $b_i = 1$ is that restriction $i$ is violated. 
  Note that a restriction $(p, q, r)_i$ is similar to a transitivity
  constraint. If $q\prec p$ and $q \prec r$, or 
  $p \prec q$ and $r\prec q$, then $b_i = 1$. Thus, for each
  restriction $(p, q, r)_i$ the following two constraints
  force $b_i=1$ if the restriction is violated.
  \begin{align*}
    b_{i} &\geq x_{qp} + x_{qr} - 1\\
    b_{i} &\geq x_{pq} + x_{rq} - 1.
  \end{align*}
  Thus, we obtain the following formulation (\refilp{ilp-naive}):
  \begin{mini!}
    {}{\sum_{(p, q, r)_i \in R}b_i\label{constr:obj}}{\label{ilp-naive}}{}
    \addConstraint{x_{pq}}{=1-x_{qp}\label{constr:asymmetry}}{\qquad\forall\, \{p, q\}\in [S]^2}
    \addConstraint{x_{pr}}{\geq x_{pq} + x_{qr} - 1\label{constr:transitivity}}{\qquad \forall\,(p, q, r) \in U}
    \addConstraint{b_{i}}{\geq x_{qp} + x_{qr} - 1\label{constr:turn-count1}}{\qquad \forall\, (p, q, r)_i \in R}
    \addConstraint{b_{i}}{\geq x_{pq} + x_{rq} - 1\label{constr:turn-count2}}{\qquad \forall\, (p, q, r)_i \in R}
    \addConstraint{x_{pq}}{\in\{0, 1\}}{\qquad\forall\, (p, q)\in S^2, p\neq q} 
    \addConstraint{b_i}{\in\{0, 1\}\label{constr:last-constr}}{\qquad\forall\, (p, q, r)_i \in R}
  \end{mini!}

  \subparagraph*{Cutting plane approach.}
  As a first improvement over \refilp{ilp-naive},  we test a cutting plane approach.
  We start by solving a relaxation of \refilp{ilp-naive} that omits
  only the transitivity constraints~\eqref{constr:transitivity}. 
  To solve the separation problem, i.e., to find violated 
  constraints~\eqref{constr:transitivity}, we search for up to $k$ cycles in the
  auxiliary graph $G'$ that contains a vertex for each location and has a 
  directed edge $(p, q)$ if and only if $x_{pq}=1$.
  A cycle in $G'$ then corresponds to a violated transitivity constraint. 
  Note that for each pair of vertices $p \neq q$ in $G'$, there is either
  an edge $(p, q)$ or an edge $(q, p)$, due to constraints~\eqref{constr:asymmetry}.
  Therefore, $G'$ is a tournament graph. 
  It is well known \cite{tournaments} that if there is a cycle in a 
  tournament graph $G'$, then there is also a cycle of length 3 in $G'$.
  Thus, it suffices to search for cycles of length $3$.

  \subparagraph*{An integer linear program via tree decompositions.}
  We now propose a formulation that reduces the number of transitivity
  constraints by exploiting the structure of the location graph.
  In particular, given a tree decomposition
  $\mathcal{T} = (T, \{X_t\}_{t \in V(T)})$ of~$\mathcal{L}$, we modify
  \refilp{ilp-naive} by using 
  $U_t = \{(p, q, r) \colon \{p, r\}\in[X_t]^2,  q\in X_t, 
  q\neq p, q\neq r\}$ instead of~$U$. Thus, we obtain the following 
  formulation (\refilp{ilp-tw}):
  \begin{mini!}
    {}{\sum_{(p, q, r)_i \in R}b_i}{\label{ilp-tw}}{}
    \addConstraint{x_{pq}}{=1-x_{qp}}{\qquad\forall\,(p, q) \in \bigcup_{t\in V(T)} [X_t]^2}
    \addConstraint{x_{pr}}{\geq x_{pq} + x_{qr} - 1}{\qquad\forall\,(p, q, r)\in\bigcup_{t\in V(T)} U_t}
    \addConstraint{b_{i}}{\geq x_{qp} + x_{qr} - 1}{\qquad\forall\,(p, q, r)_i \in R}
    \addConstraint{b_{i}}{\geq x_{pq} + x_{rq} - 1}{\qquad\forall\,(p, q, r)_i \in R}
    \addConstraint{x_{pq}}{\in\{0, 1\}\label{constr:tw-xvar}}{\qquad\forall\,t\in V(T)\,\forall (p, q) \in X_t^2, p \neq q}
    \addConstraint{b_i}{\in\{0, 1\}}{\qquad\forall\, (p, q, r)_i \in R}
  \end{mini!}
  In other words, instead of introducing transitivity constraints for every
  triplet of locations in~$U$, we restrict the transitivity constraints to
  triplets of locations that appear together in at least one bag.
  The intuition behind this formulation is the following.
  Consider a separation $(A, B)$ in $\mathcal{L}$ with a separator $A\cap B$. 
  It is possible to find a turn-optimal ordering of $\mathcal{L}$ by
  separately finding turn-optimal orderings of $\mathcal{L}[A]$ and
  $\mathcal{L}[B]$ with the property that the ordering of $A\cap B$ is
  consistent in both orderings.
  In particular, programs for $\mathcal{L}[A]$ and $\mathcal{L}[B]$
  need to share only constraints for vertices in $A\cap B$.
  Note that this intuition works only if we apply one separator to the
  location graph.  For example, if we want to solve $\mathcal{L}[A]$
  and $\mathcal{L}[B]$ recursively, we need to
  separate $\mathcal{L}$ throughout the recursion, as vertices can be 
  part of multiple separators across different recursion steps. 
  In this case, orderings of separators
  that are consistent for specific separations might be in conflict with
  each other. 
  We show that a conflict between multiple separations cannot happen
  if we use a tree decomposition for the separation of $\mathcal{L}$.

  \begin{theorem}
    Let $\mathcal{E}$ be an event graph.  If $x$ is an optimal
    solution to \refilp{ilp-tw}, then $x$ implies a turn-minimal
    ordering of the locations in $\mathcal{E}$.
  \end{theorem}
  \begin{proof}
    First, observe that if the variables of type~$x_{pq}$ indeed form
    a valid total ordering, then every possible turn is counted correctly 
    by the variables~$b_{i}$. Specifically, for every triplet $(p, q, r) \in R$,
    the variables $x_{pq}$ ($x_{qp}$) and $x_{qr}$ ($x_{rq}$) exist since
    every triplet in~$R$ forms a path of length~2 in~$\mathcal{L}$, 
    and every edge in $\mathcal{L}$ is contained in at least one bag
    due to (T\ref{prop:tree-decomp-t2}). 
    Thus, it remains to show that the variables of type~$x_{pq}$ model
    a valid total ordering. 

    Consider the directed graph~$G'$ whose vertices correspond to the
    vertices in~$\mathcal{L}$ and that has a directed edge $(p, q)$ if
    and only if~$x_{pq} = 1$.
    The underlying undirected graph of~$G'$ is a supergraph 
    of~$\mathcal{L}$, and the tree decomposition~$\mathcal{T}$ of~$\mathcal{L}$
    is also a valid tree decomposition of~$G'$. 
    Note that if~$G'$ is acyclic, then the variables
    of type~$x_{pq}$ model a valid ordering.
    Towards a proof by contradiction, suppose that this is not the case 
    and that~$G'$ does contain a cycle. Let~$C$ be a shortest cycle in~$G'$.
    If $C$ has length~3, then $C$ is a clique and due
    to~(P\ref{prop:tree-decomp-clique}), the tree~$T$ has a node $t$ such that
    $C \subseteq X_t$.  Due to constraints~(\ref{constr:transitivity}),
    however, for every bag~$X_i$ of~$\mathcal{T}$, the graph~$G'[X_i]$
    is acyclic.
    Hence, $C$ is a cycle of length at least~4.

    We claim that every bag of~$\mathcal{T}$ contains at most two
    vertices of~$C$. Otherwise, there would be two vertices~$u$
    and~$w$ of~$C$ that are not consecutive along~$C$ but, due to
    constraints~\eqref{constr:tw-xvar}, the graph~$G'$ would contain
    the edge~$(u,w)$ or the edge~$(w,u)$.  In the first case, the path
    from~$w$ to~$u$ along~$C$ plus the edge~$(u,w)$ would yield a
    directed cycle.  In the second case, the path from~$u$ to~$w$
    along~$C$ plus the edge~$(w,u)$ would also yield a directed cycle.
    In both cases, the resulting cycle would be shorter than~$C$
    (since~$u$ and~$w$ are not consecutive along~$C$).  This yields
    the desired contradiction and shows our claim.
    
    Note that~$\mathcal{T}$ restricted to~$C$ is also a tree decomposition. 
    However, since every bag contains at most two vertices of~$C$, 
    the restricted tree decomposition would have width~1, which is a 
    contradiction since every tree decomposition of a cycle has width at 
    least~2. Thus, the directed cycle~$C$ does not exist, and~$G'$ is acyclic.
  \end{proof}
  Note that this formulation requires only
  $\mathcal{O}(\operatorname{tw}(\mathcal{L})^2\cdot
  |V(\mathcal{L})|)$ variables and
  $\mathcal{O}(\operatorname{tw}(\mathcal{L})^3\cdot
  |V(\mathcal{L})|)$ many constraints.  This is a significant
  improvement over \refilp{ilp-naive} if
  $\operatorname{tw}(\mathcal{L})$ is small, which can be expected
  since train infrastructure is usually sparse and ``tree-like''.

  \section{Experimental Analysis}\label{sec:experiments}
  We tested the effectiveness of the reduction rule and the runtime of
  our formulations on an anonymized and perturbed dataset with 19 instances 
  provided by DB InfraGO AG; see \cref{fig:overview-dataset} for an overview
  of the dataset.
  Our computations show that the minimum number of turns is between 0 and 5
  in the provided dataset.
  We implemented our algorithms in the programming language \texttt{Python}.
  We used Networkx \cite{networkx} to handle most of the graph operations
  and Gurobi (version 12.0.1) \cite{gurobi} to solve the integer linear programs. 
  All experiments were conducted on a laptop running Fedora
  40 with Kernel 6.10.6 using an \mbox{Intel-7-8850U} CPU with
  four physical cores and 16\,GB RAM.

  \begin{figure}[h!]
    \centering
    \includegraphics[scale=0.9]{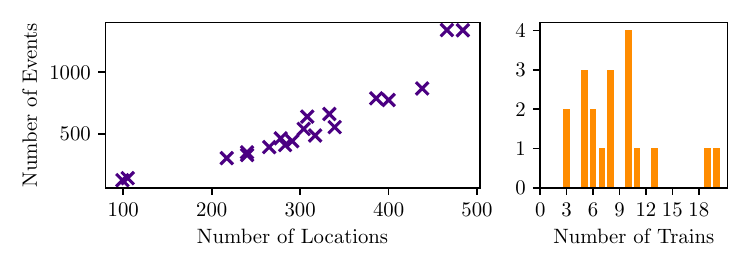}
    \caption{Left: Instances with respect to the number of events and
      the number of locations. Right: The histogram depicts the
      frequency of instances (y-axis) with a given number of trains
      (x-axis) in the dataset; e.g., there are five instances with
      five trains.}
    \label{fig:overview-dataset}
  \end{figure}

  \subparagraph{Effectiveness of the reduction rule.}
  Our implementation of \cref{rule:transit-components} is restricted to
  exhaustively contract chains. But even with this restriction, the
  reduction rule proves to be effective on the provided dataset.
  On average, the number of locations was reduced by 75\%, where the
  best result was a reduction by 89\% (instance 50-4) and the worst
  result was a reduction by 55\% (instance 20-3). A full evaluation
  is shown in \cref{fig:reduction-experimental}.
  \begin{figure}
    \centering
    \includegraphics{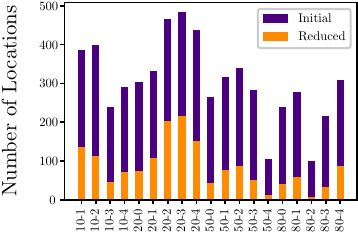}
    \caption{Results of the effectiveness of  
    applying contractions, restricted to contracting chains.
    The bar diagram shows the number of
    locations in each instance before and after contraction.}
    \label{fig:reduction-experimental}
  \end{figure}

  \subparagraph{Runtime of the ILPs.} 
  Each formulation was implemented with only one variable for each
  unordered pair of locations and without constraints~\eqref{constr:asymmetry}.
  Instead, if we created variable $x_{pq}$ for the unordered pair $\{p, q\}$ 
  and $x_{qp}$ was needed in the formulation, 
  we substituted $1-x_{pq}$ for $x_{qp}$.
  Since the computation of an optimal tree decomposition is
  \NP-hard, we used the \emph{Min-Degree Heuristic} implemented in
  Networkx to compute a tree decomposition for \refilp{ilp-tw}.
  The additional time needed for computing this tree decomposition was
  counted towards the runtime of \refilp{ilp-tw}.
  We tested the cutting plane method and \refilp{ilp-tw} on 
  the original instances which we call \emph{initial instances}
  and the instances that were reduced by our
  implementation of \cref{rule:transit-components}
  which we call \emph{reduced instances}. 
  We imposed a time limit of one hour on every experiment. 
  The result of each experiment is the average
  value over 5 repetitions of the experiment.
  The cutting plane approach on the initial instances was able
  to solve 11 out of 19 instances. An instance was either always 
  solved to optimality or never solved to optimality across all repetitions
  of the experiment.  The largest instance
  that was solved to optimality contained 465 locations (and 19 trains) and was
  solved in 3509\,s. The smallest instance that was not solved
  within the time limit contained 277 locations (and 8 trains).
  In contrast, the cutting plane approach was able to solve
  every reduced instance within 70\,s.
  Applied to the reduced instances, the cutting plane approach was 625
  times faster than when applied to the initial instances (averaged
  over those that were solved within the time limit).
  See \cref{fig:runtime-ilps} for more details on
  the performance of the cutting plane approach.
  \begin{figure}  
    \centering
    \includegraphics[scale=0.9]{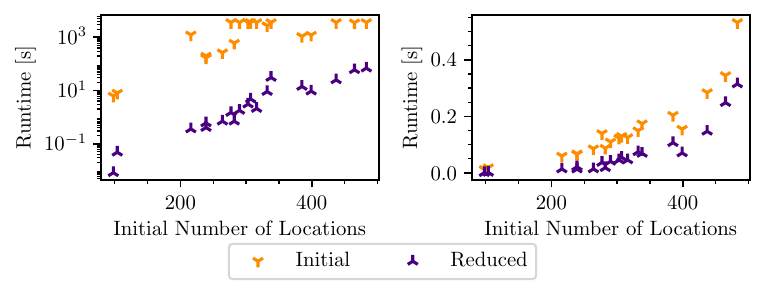}
    \caption{Runtimes of the cutting plane approach (left) and of
      \refilp{ilp-tw} (right), on the initial and the reduced
      instances.}
    \label{fig:runtime-ilps}
  \end{figure}
  \refilp{ilp-tw} solved every instance (initial or reduced) in less
  than a second (see \cref{fig:runtime-ilps} for details).  Still,
  the reduction helped, making \refilp{ilp-tw} on average 3.98 times
  faster than on the initial instances.  On the reduced instances,
  \refilp{ilp-tw} was on average 103 times faster than the cutting
  plane approach on the reduced instances (and 22,845 times faster
  than the cutting plane approach on the initial instances).
  %
  %

  \section{Conclusion and Future Work}\label{sec:conclusion}
  In this paper we have considered the problem of visualizing general
  event graphs as time-space diagrams. We established a connection between
  minimizing the number of turns in a time-space diagram and 
  \textsc{Maximum Betweenness}, we proposed a preprocessing method to
  reduce the size of event graphs, and proposed two different integer
  linear program formulations.
  
  We evaluated the performance of our algorithms on a real-world data set
  and observed that our best algorithms were able to solve every instance
  within one second. This suggests that turn-optimal time-space diagrams
  can be used in a real-time environment in practice. This can facilitate
  dispatching and the development of dispatching software.

  Future work includes evaluating the usefulness and comprehensibility of 
  the generated drawings in a real-life scenario, as well as exploring
  alternative optimization criteria that may yield improved visualizations.
  Secondary optimization steps that are applied as a post-processing 
  to a turn-optimal drawing such as minimizing the number of 
  crossings while keeping the ordering of locations might produce even
  better drawings.
  Since event graph visualizations represent time schedules, changes in the
  underlying schedule can lead to significant shifts in layout.
  An important direction for future work is developing techniques to 
  preserve the user's mental map during such updates.


  \bibliography{references}

  \clearpage
  \appendix
  \section{Complexity Results for the Turn-Minimization Problem}
  \label{sec:tstm-complexity}

  \textsc{Betweenness} is the decision version of \textsc{Maximum Betweenness}
  which asks whether an ordering exists that satisfies \emph{all} restrictions. 
  \textsc{Betweenness} was shown to be \NP-hard by Opatrny~\cite{opatrny-betweenness}.
  To prove NP-hardness of \textsc{Turn Minimization} we provide a simple reduction from 
  \textsc{Betweenness} which can be used to derive that \textsc{Turn Minimization} is also hard 
  to approximate with constant (multiplicative or additive) factor and 
  \paraNP-hard for the natural parameter, the number of turns.

  To prove the second part of the theorem, ie., that \textsc{Turn Minimization} is
  \paraNP-hard if parameterized by the vertex cover number of the
  location graph, we use a reduction from \textsc{MaxCut}.
  The decision version of \textsc{MaxCut} asks whether there is a cut of size at
  least $k$ in a given graph.
  
  \problemcomplexity*\label{thm:problem-complexity*}
  \begin{proof}
    Let $(S, R)$ be an instance of the \textsc{Betweenness} problem.
    We construct an instance for \textsc{Turn Minimization} corresponding to $(S, R)$.
    Let $S$ be the set of locations of the event graph that we want to construct. 
    In order to construct the event graph $\mathcal{E}$, we consider $R$ in an 
    arbitrary but fixed order, where we let the $i$-th element in the order be 
    denoted by $(a_i, b_i, c_i)$.
    For each triplet $(a_i, b_i, c_i) \in R$, we construct a path $\Xi_i$ in 
    $\mathcal{E}$ which is composed of three events 
    $\Xi_i = \langle v_i^1, v_i^2, v_i^3\rangle$ such that
    $\locationOf{v_i^1} = a_i$, $\locationOf{v_i^2}=b_i$, $\locationOf{v_i^3}=c_i$
    and $\timeOf{v_i^1} < \timeOf{v_i^2} < \timeOf{v_i^3}$.
    This path $\Xi_i$ can be considered as a train moving from location
    $a_i$ to $c_i$ over $b_i$.
    
    \begin{claim}
      Using this transformation, $(S, R)$ has a valid ordering $\prec$, 
      satisfying all restrictions in $R$ if and only if the transformed 
      instance $\mathcal{E}$ can be drawn as a time-space diagram $\Gamma$
      without any turns.
    \end{claim}
    \begin{claimproof}
      Let $(S, R)$ be a \textsc{Betweenness} instance with a valid ordering
      $\prec$. 
      If arranging the locations of $\mathcal{E}$ on levels from top to bottom
      according to $\prec$, for each triplet $(a_i, b_i, c_i) \in R$ it holds
      that either $a_i \prec b_i \prec c_i$ or $c_i \prec b_i \prec a_i$.
      By construction each $\Xi_i$ corresponding to a constraint 
      $(a_i, b_i, c_i) \in R$ is a path of events with three consecutive locations
      $a_i$, $b_i$, $c_i$, thus $\Xi_i$ is either monotonically increasing or
      decreasing in $\Gamma$. Therefore, no turn occurs. 
      
      Conversely, assume there is a turn-free drawing $\Gamma$ of the transformed
      instance $(\mathcal{E}, Y)$. Let~$\prec$ be the ordering of $S$ implied 
      by the mapping $y$ of $\Gamma$.

      Now, assume that $\prec$ is violates a constraint in $(a_i, b_i, c_i) \in R$.
      Thus, $b_i$ is not between $a_i$ and $c_i$ in $\prec$. 
      Then, in the corresponding path $\Xi_i$ the location $b_i$ is also not between 
      $a_i$ and $c_i$ in the mapping $y$. 
      Therefore, $y(b_i)$ is the smallest/largest level of the three levels
      $y(a_i)$, $y(b_i)$, $y(c_i)$, implying a turn in $\Gamma$; a contradiction.
    \end{claimproof}
    
    In the reduction, we have shown that we can decide \textsc{Betweenness} by 
    testing whether there are no turns in the instance for 
    \textsc{Turn Minimization} obtained in the transformation. 
    This immediately eliminates the possibility of a parameterized algorithm
    (unless $\Polytime =\NP$) whose only parameter is the number of turns 
    as we would then be able to decide \textsc{Betweenness} in polynomial time.
    Thus, we have shown \eqref{para-np-turns}.

    Approximation algorithms with a multiplicative or constant additive
    factor are also impossible (unless $\Polytime = \NP$)
    for similar reasons. 
    A multiplicative $\alpha$-approximation algorithm is ruled out by the fact
    that such an approximation algorithm would have to solve the cases with 
    $0 = 0 \cdot \alpha$ turns optimally.

    If there was an additive $\beta$-approximation algorithm for a 
    $\beta \in \operatorname{poly}(|V(\mathcal{E})|)$, we could copy each
    gadget $\beta+1$ times. 
    By duplicating the gadgets $\beta+1$ times, a turn in one gadget causes 
    all other copies of the gadget to have a turn as well. 
    Thus, the number of turns is divisible by $\beta+1$. 
    If there was an optimal mapping from locations to levels without turns, the
    additive algorithm would also have to return the optimal value, 
    since the only number in the range $\{0, \beta\}$ that is divisible by 
    $\beta +1 $ is 0, implying that we could again use this algorithm to
    decide \textsc{Betweenness}. Thus, we have shown \eqref{inapprox}.

    In order to show \eqref{para-np-vc}, we carry out a simple reduction from 
    \textsc{MaxCut}.
    Let $(G, k)$ be an instance to \textsc{MaxCut}. To transform $(G, k)$ into an instance of \textsc{Turn Minimization} we construct the following event
    graph $\mathcal{E}$.
    Let $z$ be an auxiliary vertex and let $V(G) \cup \{z\}$ be the locations
    in~$\mathcal{E}$. For each $\{u, v\} \in E(G)$, we add a train line 
    corresponding to locations $(u, z, v)$ into $\mathcal{E}$.

    \begin{claim}
      Using the transformation as described above, $(G, k)$ has a cut of size
      $k$ if and only if the transformed instance $\mathcal{E}$ can be drawn
      as a time-space diagram with $n-k$ turns.
    \end{claim}
    \begin{claimproof}
      Let $(G, k)$ be a \textsc{MaxCut} instance and let $(S, T)$ be a cut
      of size $k$ in $G$ with $C = \{\{s, t\}\colon s \in S, t\in T\}$. 
      We construct an ordering of the locations of $\mathcal{E}$ in the 
      following way.
      Every vertex in $S$ is placed below $z$ and every vertex in $T$ is placed
      above $z$ in an arbitrary order.
      Since $\mathcal{E}$ contains only train lines of the form $(u, z, v)$, for
      each $\{u, v\} \in E(G)$, this corresponds to a drawing of $\mathcal{E}$,
      where every train line corresponding to an edge in $C$ is drawn without a
      turn, and every train line whose edge is either contained in $S$ or in $T$
      is drawn with a turn. Since $(S, T)$ is a cut of size $k$, this drawing has
      $n-k$ turns.

      Conversely, let $\mathcal{E}$ be the transformed instance, let $\Gamma$
      be a drawing of $\mathcal{E}$ with $n-k$ turns and let $\prec$ be the ordering
      of locations of $\mathcal{E}$ implied by $\Gamma$.
      We set $S = \{s \in V(\mathcal{L}(\mathcal{E}))\colon s \prec z\}$ and
      $T = \{t \in V(\mathcal{L}(\mathcal{E}))\colon z \prec t\}$. The size of the
      resulting cut $(S, T)$ is $k$, which can be shown analogously to the previous 
      argument.
    \end{claimproof}
    Note that this transformation results in a location graph of $\mathcal{E}$
    that is a star graph with $z$ in the center. Since a star graph has a 
    vertex cover number of 1, 
    there is no algorithm parameterized by the vertex cover number, unless
    $\Polytime = \NP$.
  \end{proof}

  In order to show \cref{thm:fpt-tw}, we use a specific type of tree decomposition.
  We call a tree decomposition $\mathcal{T} = (T, \{X_t\}_{t \in V(T)})$ \emph{nice},
  if $T$ is rooted at a leaf node $r$, the leaf nodes in $T$ have empty bags, and all 
  other nodes are one of the three following different types.
  A node $t$ is of type \emph{introduce} if $t$ has exactly one child $c$, 
  and $X_t = X_c \cup \{v\}$ for some $v \notin X_c$. 
  Similarly, a node $t$ is of type \emph{forget}, if $t$ has exactly one child $c$,
  and $X_c = X_t \cup \{v\}$ for some $v \notin X_t$. 
  The third type is a \emph{join} node, which is a node $t$ with two children 
  $i, j \in V(T)$ whose bags contain the same vertices of $V$, i.e.,
  $X_t = X_i = X_j$. Further, we require that the root node $r$ is of type forget and
  that every leaf node in $T$ is associated with an empty bag.
  Given an arbitrary tree decomposition, a nice tree decomposition of the same graph
  can be computed in polynomial time preserving
  the width of the given decomposition such that this nice tree decomposition contains 
  $\mathcal{O}(\operatorname{tw}(G)\cdot n)$ many nodes \cite{nice-tree-decomp}.

  \paramalgo*\label{thm:fpt-tw*}
  \begin{proof}
    We begin with introducing notation. In the following we refer to  the time-space 
    diagram simply as ``drawing'' and for the sake of brevity 
    we say ``a drawing of location graph~$\mathcal{L}$'' where we 
    mean the drawing of $\mathcal{E}$ restricted to the locations contained in 
    $\mathcal{L}$.
    Given a strict total order $\prec$ on a finite set $S$, 
    the $\operatorname{rank}(b)$ of an element $b \in S$ is the position in the unique 
    enumeration of $S$ such that for each pair $a \prec b$, $a$ is enumerated before $b$. 
    Thus, $\operatorname{rank}(b) = |\{a \in S \mid  a \prec b\}| + 1$.
    
    Let $\mathcal{T} = (T, \{X_t\}_{t \in V(T)})$ be a nice tree decomposition of 
    $\mathcal{L}'$ rooted at some $r \in V(T)$.
    For some $t \in V(T)$, we define the subgraph $\mathcal{L}'_t$ of 
    $\mathcal{L}'$ to be the graph induced by the union of bags contained
    in the subtree of $T$ rooted at $t$. 
    For instance, the induced graph $\mathcal{L}'_r$ with respect to the 
    subtree rooted at the root node $r$ is precisely $\mathcal{L}'$.

    Further, let $\pi^t$ be an order of the vertices in the bag $X_t$. 
    We say that a drawing \emph{respects}~$\pi^t$ if the
    vertices in $X_t$ are drawn such that for all $p, q \in X_t$ with 
    $p \prec_{\pi_t} q$ vertex $p$ is drawn above vertex $q$ 
    (i.e., is assigned a higher level).
    With $\pi^t_{p \rightarrow i}$ we denote the order $\pi^t$ which is
    extended by a vertex $p$ such that $p$ has $\operatorname{rank}(p) = i$
    within the extended order $\pi^t_{p \rightarrow i}$ and all other
    vertices in the order that previously had a rank of at least $i$ now
    have their rank increased by 1.
    Additionally, let $\pi^t$ be an order of the vertices in the bag $X_t$ and
    let $c$ be a child of $t$ with $I = X_t \cap X_c$.
    We write $\pi^t|_{c}$ for the order $\pi^t$ restricted to the vertices in $I$.
    Lastly, we define $b(p, \pi^t)$ to be the number of turns for which 
    $p$ is one of the three locations $(p_{i-1}, p_i, p_{i+1})$ of a turn in a
    drawing of $\mathcal{L}'[X_t]$ respecting $\pi^t$. 
    Similarly, we write $b(X_t, \pi^t)$ for the total number of turns
    occurring in a drawing of $\mathcal{L}'[X_t]$ respecting the order
    $\pi^t$, where all three locations $(p_{i-1}, p_i, p_{i+1})$ of a turn
    are contained in $X_t$.
    Note that each drawing of $\mathcal{L}'[X_t]$ respecting $\pi^t$ has
    the same number of turns since $\pi^t$ dictates an ordering on every 
    vertex in $X_t$.

    Now, let $\mathcal{L}'$ be a given augmented location graph and let
    $\mathcal{T} = (T, \{X_t\}_{t \in V(T)})$ be a nice tree decomposition
    of $\mathcal{L}'$ rooted at a leaf $r \in V(T)$.
    We define $D[t, \pi^t]$ to be the number of turns in a turn-optimal
    drawing of $\mathcal{L}'_t$ respecting the order~$\pi^t$.
    Therefore, $D[r, \pi^r]$ corresponds
    to the number of turns of a turn-optimal drawing of $\mathcal{E}$,
    since $r$ is the root of $T$ and $r$ is associated with a leaf bag
    $X_r = \varnothing$. 

    We show how $D[t, \pi^t]$ can be calculated by the following recursive
    formulas depending on the node type of $t$. Based on this recursive
    formulation, the actual optimal ordering of the locations in 
    $\mathcal{E}$ can be extracted via a straightforward backtracking
    algorithm.

    \begin{description}
      \item[\textbf{Leaf node (except root):}] Since the associated bag
      $X_t$ of a leaf node $t$ is empty, 
      $\mathcal{L}'_t$ is an empty graph and therefore the minimum number
      of turns of a turn-optimal drawing of~$\mathcal{L}'_t$ is 
      $D[t, \pi^t] = 0$. 
      
      \item[\textbf{Introduce node:}] Let $p$ be the vertex that has been
      introduced in node $t$, and let $c$ be the only child of $t$, then
      \begin{align*}
          D[t, \pi^t] = D[c, \pi^t|_{c}] + b(p, \pi^t).
      \end{align*}
      Inductively, $D[c, \pi^t|_c]$ corresponds to the number of turns of a
      turn-optimal drawing of~$\mathcal{L}'_c$ respecting the order
      $\pi^t$ restricted to vertices in $X_c$. 
      The node $t$ extends the graph $\mathcal{L}'_c$ by the vertex $p$,
      introducing edges between $p$ and vertices
      $N(p) \cap X_t$ that can cause turns including~$p$. These turns are
      counted by $b(p, \pi^t)$.
      Since $\pi^t$ dictates the relative position of every vertex in
      $N[p]$, a turn-optimal drawing of~$\mathcal{L}'_t$ respecting
      $\pi^t$ must contain every newly introduced turn.
      
      Note that we count a turn at most once in this setting: First, the
      vertex $p$ is introduced exactly once in $\mathcal{L}'_t$ by the
      definition of a nice tree decomposition.
      Further, by the definition of $b$, we count a turn only if one
      location of its consecutive events $v_{i-1}, v_i, v_{i+1}$ is $p$.
      Therefore, during the computation of $D[c, \pi^t]$, the turns
      involving $p$ have not been counted previously.

      \item[\textbf{Forget node:}] Let $X_t = X_c \setminus \{p\}$ be the
      bag of $t$, where $p$ is the vertex that has been forgotten in node
      $t$ and where $c$ is the only child of $t$, then 
      \begin{align*}
        D[t, \pi^t] = \min\{D[c, \pi^t_{p \rightarrow i}]\colon 
        i = 1, \dots, |X_c|\}.
      \end{align*}
      At node $t$, we remove vertex $p$ from the bag $X_c$, therefore
      $\mathcal{L}'_t = \mathcal{L}'_c$. 
      The ordering~$\pi^t$ dictates the drawing for $\mathcal{L}'[X_c]$
      in a turn-optimal drawing in $\mathcal{L}'_t$ except for $p$.
      Thus, the number of turns of a turn-optimal drawing of 
      $\mathcal{L}'_t$ respecting $\pi^t$ must be a turn-optimal drawing
      in $\mathcal{L}'_c$ respecting the order $\pi^t$, where $p$ is 
      inserted into the order~$\pi^t$ for some $\operatorname{rank}(p) = i$.
      
      Note that every turn involved in a turn-optimal drawing respecting 
      $\pi^t_{p\rightarrow i}$ for an optimal $i$ is accounted for 
      precisely once since $p$ be can forgotten only once. Since~$p$ was
      forgotten,~$X_t$ is a separating set that separates $p$ from every 
      other vertex in~$\mathcal{L}'\setminus \mathcal{L}'_t$, implying 
      that the neighbourhood of $p$ was already processed in
      $\mathcal{L}'_c$. Further, since the locations of every possible
      turn are a triangle in $\mathcal{L}'$, we know that there is an
      already processed bag that contains all three locations of 
      consecutive events $v_{i-1}$, $v_i$, $v_{i-1}$ that can cause a turn.

      \item[\textbf{Join node:}] Let $i$ and $j$ be the two children of
      node $t$, then we can calculate the number of turns in a drawing
      of $G_t$ respecting $\pi^t$ by
      \begin{align*}
          D[t, \pi^t] = D[i, \pi^t] + D[j, \pi^t] - b(X_t, \pi^t).
      \end{align*}
      At a join node, two independent connected components of $T$ are
      joined, where~$\mathcal{L}'_i$ and $\mathcal{L}'_j$ have only 
      vertices in~$X_t$ in common. 
      By induction, $D[i, \pi^t]$ and $D[j, \pi^t]$ contain
      the number of turns in a turn-optimal drawing in~$\mathcal{L}'_i$
      and a turn-optimal drawing in $\mathcal{L}'_j$, where both drawings
      respect $\pi^t$. Consequently, by summing the number of turns in
      both drawings~$\mathcal{L}'_i$ and~$\mathcal{L}'_j$, we count turns 
      occurring in $X_t$ twice. Therefore, we need to subtract turns
      whose three corresponding vertices are contained in $X_t$.
      Further, note that no new vertex is introduced in a join node, 
      thus no new turn can occur.
    \end{description}

    With the description of the recursive formulation of $D[t, \pi^t]$,
    we have shown that a turn $(p_{i-1}, p_i, p_{i+1})$ in a turn-optimal
    drawing is counted at least once in an introduce node of the last 
    introduced vertex $p$ of $(p_{i-1}, p_i, p_{i+1})$. 
    We have also argued in the description of the forget node that a
    turn at $p$ is counted at most once. Therefore, we count every turn
    exactly once in a drawing calculated by $D[r, \varnothing]$, 
    concluding the correctness of the algorithm.

    As for the runtime, note that $b(p, \pi^t)$ can be computed in
    $\mathcal{O}(|X_t|^2)$ time by annotating every clique $\{p, q, r\}$
    in $\mathcal{L}'$ corresponding to a potential turn by the number of
    distinct consecutive event triplets mapping to locations $p$, $q$,
    and $r$. Since $p$ is involved in each clique, we can enumerate
    every clique $\{p, q, r\}$ in $\mathcal{O}(|X_t|^2)$ time and due to
    the annotation, a clique can be processed in constant time.
    In order to count the total number of turns~$b(X_t, \pi^t)$ in a bag
    $X_t$, we need $\mathcal{O}(|X_t|^3)$ time.
    Assuming the algorithm operates on a nice tree decomposition
    $\mathcal{T}$ of width $\operatorname{tw}(\mathcal{L}')$,
    there are $\mathcal{O}(\operatorname{tw}(\mathcal{L}')\cdot n)$ many
    bags in $\mathcal{T}$. 
    For a node $t$ in the tree decomposition, we have to guess 
    $\mathcal{O}((\operatorname{tw}(\mathcal{L}')+1)!)$ many orders.
    If $t$ is a leaf node, it can be processed in constant time. Given an 
    order~$\pi^t$, we can compute any introduce, forget, or join node in
    $\mathcal{O}(\operatorname{tw}(\mathcal{L}')^3)$ time, yielding
    an overall runtime of 
    $\mathcal{O}((\operatorname{tw}(\mathcal{L}')+1)! \cdot 
    \operatorname{tw}(\mathcal{L}')^4 \cdot n)$. 
  \end{proof}
\end{document}